\documentclass[11pt,letterpaper,oneside]{article}
\usepackage{hyperref}
\usepackage{latexsym}
\usepackage{amsmath}
\usepackage{amssymb}
\usepackage{amsthm}
\usepackage{xy}
\usepackage{url}
\usepackage{graphicx}
\usepackage[left=0.75in, right=0.75in, top=1in, bottom=1in]{geometry}
\usepackage{caption}
\usepackage{subcaption}
\usepackage{xspace}

\theoremstyle{definition} 
\newtheorem{definition}{Definition}[section]

\newtheorem{question}[definition]{Question}

\newtheorem{remark}[definition]{Remark}

\newtheorem*{definition*}{Definition}

\newenvironment{proof-idea}{\noindent \textit{Proof idea.}}{\nobreak\hfill\qed\bigskip}

\theoremstyle{plain}

\newtheorem{theorem}[definition]{Theorem}
\newtheorem{conjecture}[definition]{Conjecture}

\newtheorem*{proposition*}{Proposition}
\newtheorem*{lemma*}{Lemma}
\newtheorem*{corollary*}{Corollary}
\newtheorem*{theorem*}{Theorem}
\newtheorem*{conjecture*}{Conjecture}

\input xy
\xyoption{all}

\newcommand{\definedWord}[1]{\emph{#1}}
\newcommand{\ie}{i.\,e.}
\newcommand{\eg}{e.\,g.}

\newcommand{\cc}[1]{\ensuremath{\mathsf{#1}}} 

\DeclareGraphicsExtensions{.png}

\DeclareMathOperator{\poly}{poly}

\begin{document}

\title{$\cc{NP}$-hard sets are not sparse unless $\cc{P} = \cc{NP}$: \\
An exposition of a simple proof of Mahaney's Theorem, with applications}

\author{Joshua A. Grochow\footnote{Department of Computer Science, University of Colorado at Boulder, \texttt{joshua.grochow@colorado.edu} and Santa Fe Institute, \texttt{jgrochow@santafe.edu}}}

\maketitle

\begin{abstract}
Mahaney's Theorem states that, assuming $\cc{P} \neq \cc{NP}$, no $\cc{NP}$-hard set can have a polynomially bounded number of yes-instances at each input length. We give an exposition of a very simple unpublished proof of Manindra Agrawal whose ideas appear in Agrawal--Arvind (``Geometric sets of low information content,'' Theoret. Comp. Sci., 1996). 
This proof is so simple that it can easily be taught to undergraduates or a general graduate CS audience---not just theorists!---in about 10 minutes, which the author has done successfully several times. We also include applications of Mahaney's Theorem to fundamental questions that bright undergraduates would ask which could be used to fill the remaining hour of a lecture, as well as an application (due to Ikenmeyer, Mulmuley, and Walter, arXiv:1507.02955) to the representation theory of the symmetric group and the Geometric Complexity Theory Program. To this author, the fact that sparsity results on $\cc{NP}$-complete sets have an application to classical questions in representation theory says that they are not only a gem of classical theoretical computer science, but indeed a gem of mathematics.
\end{abstract}

\pagestyle{myheadings}
\markboth{A simple proof of Mahaney's Theorem - J. A. Grochow}{A simple proof of Mahaney's Theorem - J. A. Grochow}

\section{Introduction}
Mahaney's Theorem \cite{mahaney} is one of the seminal results in the pre-probabilistic era of computational complexity, and answers several foundational questions about the nature of the $\cc{P}$ versus $\cc{NP}$ question. The theorem states that there are no sparse $\cc{NP}$-complete sets unless $\cc{P}=\cc{NP}$. Recall that a set $L \subseteq \Sigma^{*}$ is (polynomially) sparse if the number of strings of length $n$ in $L$ is bounded by a polynomial in $n$.

Unfortunately, this theorem is much less well-known than it should be, partially due to the fact that Mahaney's original proof---using the so-called ``census technique'' originated by Fortune \cite{fortune}---although not complicated, takes at least one full lecture to teach. Given the myriad topics---most of them more modern and in vogue---which teachers feel they must cover in undergraduate and graduate courses on complexity, it's unclear if the census technique really deserves a whole lecture (or more), especially since the applications of the census technique more or less stop at Mahaney's Theorem. On looking at the topics often taught in complexity classes nowadays, the general consensus seems to be that it is not worth a whole lecture in lieu of other topics.

In this note, we give an exposition of a much simpler proof, due to Manindra Agrawal (unpublished, but the idea of the proof is present in Agrawal and Arvind \cite{agrawalArvind}, where the technique was put towards proving a stronger statement). This proof is so simple that it can be taught to undergraduates or a general graduate CS audience---not just theorists!---in about 10 minutes, which the author has done successfully several times. (Not to mention the value of having multiple and simpler proofs of important theorems.) The rest of a lecture could then be spent on other topics, or, I suggest, on applications of Mahaney's Theorem to fundamental questions about $\cc{P}$ versus $\cc{NP}$ that a bright undergraduate would ask (which I cover in Section~\ref{sec:applications}). Alternatively, these applications could potentially be given as homework after having taken the 10 minutes to teach Mahaney's Theorem. We'll cover more of the history in Section~\ref{sec:applications} in the course of discussing the applications of Mahaney's Theorem, but for now let's jump right into the proof. We also note that slight variants of Agrawal's proof can be used to give similarly simple proof of Fortune's Theorem---no $\cc{coNP}$-complete set is sparse unless $\cc{P} = \cc{NP}$ (see Remark~\ref{rmk:fortune})---and a common generalization of Mahaney's and Fortune's Theorems to a slightly more general kind of reduction (see Question~\ref{question:pclose}). The proof is reminiscent of the ``left set'' technique of Ogiwara--Watanabe \cite{ogiwaraWatanabe} (which they used to generalize Mahaney's Theorem to $\leq_{btt}^p$ reductions: non-adaptive Turing reductions that make only $O(1)$ queries), but when applied to the case of Mahaney's Theorem, Agrawal's proof is the simplest proof of which the author is aware (at the very least, requiring the fewest additional definitions and lemmas).

\subsection*{Acknowledgments}
We gratefully acknowledge Manindra Agrawal for allowing us to publish this proof. Certainly I am not the only person aside from Agrawal to have been aware of this proof; at the very least Thomas Thierauf was also aware, and was very supportive of writing it up. We thank Thierauf for his help, support, and suggestions. We also thank Lance Fortnow for useful comments on a draft.

\section{The proof}
\begin{theorem}[Mahaney \cite{mahaney}]
If $\cc{P} \neq \cc{NP}$, then no $\cc{NP}$-hard set is sparse.
\end{theorem}

\begin{proof}[Proof (M. Agrawal, compare Agrawal--Arvind \cite{agrawalArvind})]
Suppose $L$ is a sparse $\cc{NP}$-hard language. We will show how to solve SAT in polynomial time. Suppose we wish to decide whether a given propositional formula $\varphi(x_1, \dotsc, x_n)$ is satisfiable. Consider the downward self-reduction tree of $\varphi$: in its first layer, we get the two formulas $\varphi_{0} = \varphi(0, x_2, \dotsc, x_n)$ and $\varphi_1 = \varphi(1, x_2, \dotsc, x_n)$. Each level $\ell$ of the tree has the crucial property that 
\begin{equation} \label{eqn:property}
\varphi \text{ is satisfiable } \Leftrightarrow \text{ at least one formula at level $\ell$ is satisfiable}.
\end{equation}

The problem with using the downward self-reduction tree to solve SAT is that it takes exponential time, since it essentially amounts to trying all $2^n$ possible assignments. However, if we could prune the tree at each level in polynomial time in such a way as to preserve property (\ref{eqn:property}), but only have $\poly(n)$ many strings at each level, then the resulting SAT algorithm would work in polynomial time, since there are only $n$ levels. This is exactly what we will show how to do.

Let $f\colon \Sigma^{*} \to \Sigma^{*}$ be a polynomial-time many-one reduction from SAT to $L$. Let $s(n)$ (``$s$'' for ``sparsity'') be a polynomial bounding the number of strings in $L$ of length $\leq n$, and let $r(n)$ (``$r$'' for ``reduction'') be a polynomial bound on the length-stretch of $f$, that is, $|f(\varphi)| \leq r(|\varphi|)$ for all $\varphi$. We will prune the downward self-reduction tree to have at most $t(n)+1$ formulas at each level (``$t$'' for ``tree''), where $t(n) = s(r(2n+5))$, which is polynomially bounded. (We use $2n+5$ because we will apply $f$ to strings of the form $(\varphi) \vee (\psi)$ where $|\varphi|, |\psi| \leq n$, and we had to add the 5 symbols ``$()\vee()$'' to combine these two formulae together. For clarity below, we will not include these extra parentheses.)

The construction proceeds in stages, starting with the first stage at the first level of the tree. Let $\varphi_1, \dotsc, \varphi_k$ be the formulas at the current level of the tree at the current stage. If $k \leq t(n)+1$ we continue to the next level; otherwise, $k > t(n) + 1$. For $i=2, \dotsc, k$, let $q_i = f(\varphi_1 \vee \varphi_i)$, where as usual ``$\vee$'' denotes ``OR;'' note that we do not include $i=1$. 
After each stage, the remaining $\varphi_i$ at the current level of the tree are re-indexed to start from $1$ again.

\textbf{Case 1:} All $q_i$ are distinct strings. Remove $\varphi_1$ from the tree and continue to the next stage. In this case, $\varphi_1$ must have been unsatisfiable, so removing it is okay, \ie, preserves property (\ref{eqn:property}): since there are strictly more than $s(r(2n+5))$ distinct $q_i$, at least one of them must not be in $L \cap \Sigma^{\leq r(2n+5)}$, say $q_i$. Since $q_i = f(\varphi_1 \vee \varphi_i) \notin L$, $\varphi_1 \vee \varphi_i$ is not satisfiable, and in particular $\varphi_1$ is not satisfiable. (Note that we do not know, and we do not \emph{need} to know, which $q_i$ is not in $L$.)

\textbf{Case 2:} $q_i = q_j$ for some distinct $i,j$. Remove $\varphi_i$ from the tree and continue to the next stage. Let's see why this is okay. If $\varphi_1$ is satisfiable, then removing $\varphi_i$ preserves property (\ref{eqn:property}) (recall $i \neq 1$). If $\varphi_1$ is not satisfiable, then $\varphi_i$ is satisfiable if and only if $\varphi_1 \vee \varphi_i$ is satisfiable, if and only if $\varphi_1 \vee \varphi_j$ is satisfiable  (since $f(\varphi_1 \vee \varphi_i) = f(\varphi_1 \vee \varphi_j)$), if and only if $\varphi_j$ is satisfiable. So again, removing $\varphi_i$ preserves property (\ref{eqn:property}).
\end{proof}

\begin{remark} \label{rmk:fortune}
One can give a similarly simple proof of Fortune's Theorem \cite{fortune}, which says that no $\cc{coNP}$-hard set is sparse unless $\cc{P} = \cc{NP}$, or equivalently that no $\cc{NP}$-hard set is co-sparse unless $\cc{P} = \cc{NP}$. Here is how the proof should be modified: consider the language TAUT of propositional tautologies instead of SAT. Replace property (\ref{eqn:property}) by ``$\varphi$ is a tautology if and only if \emph{every} formula at level $\ell$ is a tautology.'' Use $\varphi_1 \wedge \varphi_i$ instead of $\varphi_1 \vee \varphi_i$. Proceed as before. In Case 1, we can simply stop, since in this case $\varphi_1$ is not a tautology and therefore $\varphi$ is not a tautology. In Case 2, we remove $\varphi_i$ as before. If $\varphi_1$ is not a tautology, then removing $\varphi_i$ leaves this non-tautology in the tree so we're fine. If $\varphi_1$ is a tautology, then $\varphi_i \in \text{TAUT} \Leftrightarrow \varphi_1 \wedge \varphi_i \in \text{TAUT} \Leftrightarrow \varphi_1 \wedge \varphi_j \in \text{TAUT} \Leftrightarrow \varphi_j \in \text{TAUT}$. 
\end{remark}

\begin{remark} \label{rmk:tradeoff}
From the proof, we see that Mahaney's Theorem---and Fortune's Theorem, and a common generalization as in Question~\ref{question:pclose}---has a smooth trade-off: if there is an $\cc{NP}$-hard language with at most $s(n)$ strings of length $\leq n$, then $\cc{NP} \subseteq \cc{TIME}(\poly(s(n)))$. In particular, if the Exponential Time Hypothesis \cite{impagliazzoPaturi} holds, then $\cc{NP}$-hard sets and $\cc{coNP}$-hard sets both have exponentially many strings of length $\leq n$, for all $n$.
\end{remark}

\section{Some applications} \label{sec:applications}
These applications are covered elsewhere, but for the sake of being self-contained and convincing the reader to actually learn and/or teach Mahaney's Theorem and Agrawal's proof thereof, I highlight them here.

\subsection{Foundational questions about computational complexity}
I present these first two applications in a question-and-answer format, as these are questions a bright undergraduate might ask when first learning about $\cc{P}$ versus $\cc{NP}$. Although these questions and their answers are covered in Chapter 4 of Sch\"{o}ning's book \cite{schoningBook}, the proofs there are slightly different. The first question is answered directly by Mahaney's Theorem; the second question is answered by a very slight variant of Agrawal's proof (quite distinct from the proof given in Sch\"{o}ning \cite[Chapter~4]{schoningBook}).

\begin{question}
Isn't it possible that $\cc{P} \neq \cc{NP}$, but that there is an algorithm that correctly solves SAT on all instances, but only runs in polynomial time on ``most'' instances, say, all but polynomially many?
\end{question}

This kind of question was first studied by Meyer and Paterson \cite{meyerPaterson}, who referred to an algorithm that runs in polynomial time on all instances except a sparse set as \definedWord{almost polynomial-time} or \definedWord{APT} (not to be confused with the more modern complexity class $\cc{AlmostP}$, which happens to equal $\cc{BPP}$).

\begin{proof}[Answer]
No. We prove the contrapositive. Suppose there is an almost polynomial-time algorithm $A$ that solves SAT. Let $t(n)$ be a polynomial bounding the runtime of $A$ on all strings not in the sparse set $N \subset \Sigma^*$. Without loss of generality, we may assume that for any $x \in N$, $A(x)$ takes strictly more than $t(|x|)$ time, otherwise we could have excluded $x$ from $N$ in the first place. Let $N' = N \cap \text{SAT}$; we will show that if $N'$ is empty then SAT is in $\cc{P}$, and otherwise $N'$ is $\cc{NP}$-hard (in fact, $\cc{NP}$-complete, but we won't need that). In this latter case, since $N$ was sparse, so is $N'$, and hence by Mahaney's Theorem we can again conclude that $\cc{P} = \cc{NP}$.

If $N'$ is empty, then we can directly show that SAT is in $\cc{P}$: run $A(x)$ for $t(|x|)$ steps. If it has halted, output whatever $A(x)$ did. If it hasn't halted, reject. Hence, if $N'$ is empty, then $\cc{P} = \cc{NP}$.

Otherwise, we show that $N'$ is $\cc{NP}$-hard. The following is a reduction from SAT to $N'$: let $x_{yes}$ be a fixed element in $N'$---which exists since $N'$ is nonempty---and $x_{no}$ a fixed element outside of $N'$---which exists since $N'$ is sparse. On input $x$, run $A(x)$ for $t(|x|)$ steps. If $A(x)$ has accepted by this time, output $x_{yes}$; if $A(x)$ has rejected by this time, output $x_{no}$. Otherwise, simply output $x$. It is easily verified that this is a reduction $\text{SAT} \leq_{m}^{p} N'$. Thus, whether $N'$ is empty or nonempty, we find that $\cc{P} = \cc{NP}$.
\end{proof}

\begin{question} \label{question:pclose}
Isn't it possible that $\cc{P} \neq \cc{NP}$, but that there is a polynomial-time algorithm that correctly solves SAT on ``most'' instances, say, all but polynomially many?
\end{question}

A language $L$ such that there is a polynomial-time algorithm that correctly solves $L$ on all but polynomially many instances is called \definedWord{$\cc{P}$-close}. To rephrase this definition in terms of sparse sets: a language $L$ is $\cc{P}$-close if and only if there is a language $L' \in \cc{P}$ such that the symmetric difference $L \Delta L'$ is sparse.

\begin{proof}[Answer]
No. We prove the contrapositive. Suppose that SAT is $\cc{P}$-close: there is a language $L \in \cc{P}$ such that the symmetric difference $S = \text{SAT} \Delta L$ is sparse. We will show that $S$ is $\cc{NP}$-hard (in fact, complete) under a slightly weaker kind of reduction than many-one. A very slight variant of Agrawal's proof from above will show that this is still enough to conclude $\cc{P} = \cc{NP}$. 

First we show there is a one-query polynomial-time Turing reduction from SAT to $S$. On input $x$, first compute $L(x)$ in deterministic polynomial time. If $L(x) = 0$, then $x \in \text{SAT} \Leftrightarrow x \in S$, so we make the query $x$ to $S$ and output the corresponding answer. If $L(x) = 1$, then $x \in \text{SAT}$ if and only if $x$ is \emph{not} in $S$, so we again make the query $x$ to $S$, but now we reverse the answer.

Now we show how to modify Agrawal's proof to show that if there is a sparse $\cc{NP}$-hard set under the preceding type of reduction (one-query polynomial-time Turing, also called ``1-truth table''), then $\cc{P} = \cc{NP}$. We can think of such a reduction as follows: there is a polynomial-time function $f\colon \Sigma^* \to \Sigma^*$ which determines which query to make, and a polynomial-time function $g\colon \Sigma^{*} \to \{\texttt{orig}, \texttt{neg}\}$ which determines whether the answer to the query is the answer to the original question or its negation. Proceed as in Agrawal's proof, with the following modifications (we follow the notation from the proof above). Rather than pruning down to $t(n)+1$ formulas at each level, we'll prune down to $2(t(n)+1)$. Instead of only considering $q_i = f(\varphi_1 \vee \varphi_i)$, we consider the pair $(q_i, b_i) = (f(\varphi_1 \vee \varphi_i), g(\varphi_1 \vee \varphi_i))$. 

\textbf{Case 1: } All $q_i$ are distinct. Then there are either at least $t(n)+1$ indices $i$ such that $b_i=\texttt{orig}$ or $t(n)+1$ indices $i$ such that $b_i = \texttt{neg}$. In the former case, we remove $\varphi_1$ as in the original proof, with the same proof of correctness. In the latter case, we have at least $t(n)+1$ indices $i$ such that $b_i = \texttt{neg}$. Since there are at least $t(n)+1$ such indices and all the $q_i$ are distinct, at least one of the $q_i$ with $b_i = \texttt{neg}$ is not in $L$. But this means exactly that $\varphi_1 \vee \varphi_i$ is satisfiable, and hence that the original formula $\varphi$ is satisfiable, so we stop and accept.

\textbf{Case 2:} Some $q_i = q_j$ for distinct $i,j$. If $b_i = b_j$, then we remove $\varphi_i$ as in the original proof, which works for essentially the same reason as before. If $b_i \neq b_j$, then we immediately accept. For without loss of generality, suppose $b_i = \texttt{orig}$ and $b_j = \texttt{neg}$, and let $q = q_i = q_j$. Either $q \in L$, in which case $(q_i, b_i)$ yields a positive answer, or $q \notin L$, in which case $(q_j, b_j)$ yields a positive answer. Either way, $\varphi$ is satisfiable.
\end{proof}

\subsection{History and original motivation: the Berman--Hartmanis Isomorphism Conjecture}
When we write down an $\cc{NP}$-complete set, we typically talk in terms of natural mathematical objects such as graphs, formulas, etc. But of course, the theory of computation demands that these objects be encoded in some way into strings over a finite alphabet $\Sigma$. Typically we sweep this encoding process under the rug, saying that any natural and reasonable encoding will do. But technically, even amongst natural and reasonable encodings, different encodings of, say, SAT, yield \emph{different languages} as subsets of $\Sigma^{*}$. However, these languages are isomorphic, in the sense that there is a polynomial-time computable bijection $f\colon \Sigma^* \to \Sigma^*$ with polynomial-time computable inverse $f^{-1}$ which sends one encoding of SAT to another. Two subsets of $\Sigma^*$ that are the same up to a polynomial-time computable and polynomial-time invertible bijection of $\Sigma^*$ are called \definedWord{(p-)isomorphic}. Thus, isomorphic languages are essentially just re-encodings of one another.

Berman and Hartmanis \cite{bermanHartmanis} conjectured that all $\cc{NP}$-complete sets are isomorphic---that is, they are essentially all just re-encodings of SAT. They showed some relatively mild conditions under which an $\cc{NP}$-complete set could be shown isomorphic to SAT, and all known natural $\cc{NP}$-complete sets satisfy these conditions, and hence are isomorphic. (There are a few candidate exceptions \cite{josephYoung, KMR}, but they were all constructed explicitly for the purpose of refuting the conjecture.) This is really quite remarkable, as it says that, for example, the set of knots of genus at most $g$, the set of Hamiltonian graphs, the set of solvable Super Mario Brothers levels \cite{mario}, and the set of quadratic Diophantine equations with positive solutions are all not merely computationally equivalent to SAT, but simply represent different choices of encoding satisfiable formulas into Boolean strings!

The Berman--Hartmanis Isomorphism Conjecture was central to work on computational complexity in the early 1980s. One of the properties that they pointed out is preserved by isomorphism is the density of a language---the number of strings up to length $n$, up to polynomial transformations. It was thus natural for them to ask whether there existed an $\cc{NP}$-complete problem that was sparse, which would therefore refute the Isomorphism Conjecture. Mahaney's Theorem was originally motivated by this question, and gave as strong as possible a negative answer. This string of ideas is also covered in the survey article by Hartmanis and Mahaney \cite{hartmanisMahaney}.

Sparse sets were in fact one of the central concepts in computational complexity in the 1980s, though I hope to convince you of their continued relevance and value today. In addition to the motivation from the Isomorphism Conjecture, sparse sets also became central because of their relation to circuit complexity. In particular, a language is in $\cc{P/poly}$ if and only if it is polynomial-time Turing reducible to some sparse set (due to Meyer, cited in Berman and Hartmanis \cite{bermanHartmanis}). From this perspective, the famous Karp--Lipton Theorem---$\cc{NP} \subseteq \cc{P/poly}$ implies $\cc{PH} = \cc{\Sigma_2 P}$---can be seen as a statement about the consequences of a sparse set being $\cc{NP}$-hard under polynomial-time Turing reductions. Improving the conclusion of Karp--Lipton to $\cc{P} = \cc{NP}$---or equivalently, improving the reductions in Mahaney's Theorem to polynomial-time Turing reductions---requires nonrelativizing techniques \cite{immermanMahaney} (the same is even true for strengthening Karp--Lipton to conclude $\cc{PH} = \cc{P}^{\cc{NP}}$ \cite{heller}). However, Mahaney's Theorem \emph{has} been extended from polynomial-time many-one reductions to ``bounded truth-table'' reductions: nonadapative Turing reductions that make only $O(1)$ queries \cite{ogiwaraWatanabe}, and even slightly beyond \cite{AKM}. For more on the history and surveys of related results, see, \eg, \cite{HOW, young1, young2, clarity1, clarity2, caiOgihara} and references therein.

\subsection{Representations of the symmetric group and geometric complexity theory}
We conclude with a much more recent application \cite{IMW}, which is a statement about the representation theory of the symmetric groups, which was, in fact, not known prior to this application of a sparsity theorem. Although fully understanding the technical details here would require reading some other papers and having some knowledge of representation theory, the value of Mahaney's Theorem for representation theory and geometric complexity theory (GCT) can be understood without such knowledge. We try to give an exposition of this application to make such appreciation possible, without knowing the relevant representation theory. (The relevant representation theory can be found in, \eg, the first part of Fulton and Harris \cite{fultonHarris}, Fulton's book on Young tableaux \cite{fultonYoungTableaux}, or James's book \cite{james}.)

The GCT Program (see, \eg, \cite{gct1,gct2,gct6,gctJACM,gctCACM} and references therein and thereto) proposes to resolve major conjectures like Permanent versus Determinant ($\cc{VP}_{ws}$ versus $\cc{VNP}$) or $\cc{P}$ versus $\cc{NP}$ by associating to each complexity class certain integer multiplicities coming from representation theory. The connection they made to lower bounds was the following:

\begin{theorem}[Mulmuley and Sohoni \cite{gct1}]
If there is some representation-theoretic multiplicity which is larger for the permanent than for the determinant, then $\cc{VP}_{ws} \neq \cc{VNP}$.
\end{theorem}

They also made the following conjecture, which is formally stronger than what is needed to use the preceding result:

\begin{conjecture}[Mulmuley and Sohoni \cite{gct1}]
There exist representation-theoretic multiplicities that are \emph{zero} for $\cc{VP}_{ws}$ and nonzero for $\cc{VNP}$.
\end{conjecture}

Although this conjecture was recently disproven \cite{occurrence}, the following result is still a nice application of computational complexity to representation theory.\footnote{For historical purposes, we note that before the recent disproof of the conjecture, analogous and related results (see, \eg, Ikenmeyer \cite[Section~6.3]{ikenmeyerThesis}) as well as computer experiments (also by Ikenmeyer, summarized succinctly in, \eg, Landsberg \cite[Section~6.6]{landsbergGCTsurvey}) suggested that only very few multiplicities associated to $\cc{VP}_{ws}$ were zero, which would have made it very difficult to prove the above conjecture. Of course, it's expected to be difficult given the nature of the questions it is approaching, but if too many of these multiplicities were zero, it might have made them \emph{very} hard to find. This result suggested that, despite the numerical evidence, in fact lots of these multiplicities were zero. The recent disproof of the conjecture shows that the associated multiplicities for $\cc{VNP}$ are also zero.}
The so-called ``Kronecker coefficients'' are some of the relevant representation-theoretic multiplicities associated to $\cc{VP}_{ws}$ or the determinant (they are, more precisely, an upper bound on the multiplicities associated to the determinant). The Kronecker coefficients are quite classical mathematical objects, going back more than 100 years. \textbf{For those who know a little representation theory: } more precisely, the Kronecker coefficients are the tensor product multiplicities for the symmetric group $S_n$. In other words, given two irreducible representations $V_{\lambda}$, $V_{\mu}$ of $S_n$, specified by partitions $\lambda,\mu$ of $n$, the Kronecker coefficient $k_{\lambda,\mu,\nu}$ is the multiplicity of the irreducible representation $V_{\nu}$ in the decomposition of $V_{\lambda} \otimes V_{\mu}$ into a direct sum of irreducible representations.

\begin{theorem}[Ikenmeyer, Mulmuley, Walter \cite{IMW}]
If $\cc{P} \neq \cc{NP}$, then the Kronecker coefficients are zero super-polynomially often, and hence, so are the multiplicities associated to the determinant and $\cc{VP}_{ws}$.
\end{theorem}

In fact, they show the conclusion of this theorem unconditionally, but they do so by showing that deciding the positivity of Kronecker coefficients is $\cc{NP}$-hard, and then reducing from an $\cc{NP}$-complete problem that is known to have exponentially many no-instances \cite[Lemma~5.1 and Theorem~5.2]{IMW}.

\begin{proof}
B\"{u}rgisser and Ikenmeyer \cite{burgisserIkenmeyerMM} showed that there is an $\cc{NP}$-complete problem $L$, whose naturally associated $\cc{\# P}$-complete counting problem is an upper bound on the Kronecker coefficients. In particular, the number of no-instances of $L$ at length $n$ is a lower bound on the number of Kronecker coefficients of length $n$ that are zero. By Fortune's Theorem (see the simple proof in Remark~\ref{rmk:fortune}), if $\cc{P} \neq \cc{NP}$, then the number of no-instances of $L$ must grow super-polynomially in $n$.
\end{proof}

\begin{remark}
Because of the smooth tradeoff in Fortune's Theorem (see Remark~\ref{rmk:tradeoff}), even without the unconditional result of \cite{IMW}, from their conditional result one immediately gets that if the Exponential Time Hypothesis \cite{impagliazzoPaturi} holds, then the Kronecker coefficients are zero exponentially often (which is what they then proved unconditionally).
\end{remark}

\bibliographystyle{alphaurl}
\bibliography{mahaney}

\end{document}